\documentclass[preprint]{elsarticle}
\usepackage{amsmath,amssymb,amsbsy,amsfonts,latexsym,amsopn,amstext,amsxtra,euscript,amscd,hyperref}
\usepackage{mathptmx}      
\usepackage{fancybox}
\usepackage{tikz}
\usepackage{graphicx}
\usepackage[left=3.50cm, right=3.5cm, bottom=3.30cm]{geometry}
\newtheorem{theorem}{Theorem}
\newtheorem{lemma}[theorem]{Lemma}
\newtheorem{corollary}[theorem]{Corollary}

\newdefinition{example}{Example}
\newdefinition{definition}{Definition}
\newdefinition{remark}{Remark}
\newproof{proof}{Proof}

\begin{document}
\begin{frontmatter}
\title{Negacyclic codes over $\mathbb{Z}_4+u\mathbb{Z}_4$}
\author[ra]{Rama Krishna Bandi\corref{cor1}}
\cortext[cor1]{Corresponding author}
\ead{bandi.ramakrishna@gmail.com}
\author[ra]{Maheshanand Bhaintwal}
\address[ra]{
 Department of Mathematics,
 Indian Institute of Technology Roorkee,
 Roorkee-247667, INDIA}

\begin{abstract} In this paper, we study negacyclic codes of odd length and of length $2^k$ over the ring $R=\mathbb{Z}_4+u\mathbb{Z}_4$, $u^2=0$. We give the complete structure of negacyclic codes for both the cases. We have obtained a minimal spanning set for negacyclic codes of odd lengths over $R$. A necessary and sufficient condition for  negacyclic codes of odd lengths to be free is presented. We have determined the cardinality of negacyclic codes in each case. We have obtained the structure of the duals of negacyclic codes of length $2^k$ over $R$ and also characterized self-dual negacyclic codes of length $2^k$ over $R$.
\end{abstract}


\end{frontmatter}
\section{Introduction}
Cyclic codes are widely studied algebraic codes among all families of codes. Their structure is well known over finite fields \cite{macwilliams}. Recently codes over rings have generated a lot of interest after a breakthrough paper by Hammons et al. \cite{hammons}. When we consider codes over a local ring, a cyclic code is called a distinct roots cyclic code if the code length is relatively prime to the characteristic of the residue field. Their structure over finite chain rings is now well known \cite{norton}. They have also been studied over other rings such as $\mathbb{F}_2+u\mathbb{F}_2$, $u^2=0$, \cite{bonnecaze}; $\mathbb{F}_2+u\mathbb{F}_2+v\mathbb{F}_2+uv\mathbb{F}_2$, $u^2=v^2=0, uv=vu$, \cite{yildiz1}; $\mathbb{F}_2+v\mathbb{F}_2$, $v^2=v$, \cite{zhu}; and $\mathbb{Z}_4+u\mathbb{Z}_4$, $u^2=0$ \cite{yildiz, rk}. Via the Gray map, these new families of rings also lead to binary codes with large automorphism groups and in some cases new binary codes.

When the characteristic of the residue field is not relatively prime to the code length then the corresponding cyclic codes is called a repeated root cyclic code.  Castangoli \cite{castangoli} et. al. have studied these cyclic codes over finite fields. Lint \cite{van}, Tang et. al. \cite{tang}  and  Zimmermann \cite{zimmermann} have also explored these codes in their studies. Recently cyclic codes of length $2^e$ over $\mathbb{Z}_4$ have been considered by Abualrab  and Oehmke in \cite{taher, taher1}. Blackford \cite{blackford} has classified cyclic codes of oddly even length (length $2n$, $n$ is odd) over $\mathbb{Z}_4$ using Discrete Fourier Transform. Dougherty and Ling \cite{dougherty} have generalized cyclic codes of oddly even length to any even length using the same approach.

Negacyclic codes were first introduced by Berkelamp \cite{berlekamp}. They have been generalized to $\mathbb{Z}_4$ with code length an odd integer by Wolfmann \cite{wolfmann}. Blackford \cite{blackford1} has extended the results of \cite{wolfmann} to negacyclic codes of even length, and determined all binary linear repeated root cyclic codes that are Gray images of quaternary codes. Dinh and Lopez-Permouth \cite{dinh} have studied negacyclic codes of odd length in the more general setting of finite chain rings, and also considered repeated root cyclic codes of length $2^s$ over $\mathbb{Z}_{2^m}$. The structure of negacyclic codes of length $2^s$ over Galois rings and their complete Hamming distances were  discussed by Dinh in \cite{dinh1, dinh2}. Constacylic codes of length $2^s$ and $p^s$ over Galois extension of $\mathbb{F}_2+u\mathbb{F}_2$, $u^2=0$ and $\mathbb{F}_{p^s}+u\mathbb{F}_{p^s}$, $u^2=0$, respectively, have been studied by Dinh in \cite{dinh3, dinh4}.

 Recently, Yildiz and Karadeniz \cite{yildiz} have studied linear codes over $R$.
 The authors \cite{rk} have studied cyclic codes of odd length and length $2^k$ over $R$. The ring $\frac{R[x]}{\left \langle x^{2^k}-1  \right \rangle}$ is a local ring with the unique maximal ideal $M$ with 3-generators $2, ~u,~ x-1$ (i.e. $M=\left \langle 2,~u,~x-1 \right \rangle$) \cite{rk}. The presence of $3$ generators  makes the characterization of cyclic codes of length $2^k$ over $R$ complicated, as was observed in \cite{rk}. However, the rig $\frac{R[x]}{\left \langle x^{2^k}+1  \right \rangle}$, which is also a local ring, has the unique maximal ideal with $2$ generators only (See Thorem \ref{local}). This motivated us to study negacyclic codes of length $2^k$ over $R$.

The paper is organized as follows: In Section II, we present the preliminaries such as the ring structure of $R$, Hensel's lemma, Gray map on $R$, etc. We discuss the Galois extension of $R$ and its ideal structure. In Section III, we have presented a structure of negacyclic codes of odd length over $R$, and determined their  minimal spanning sets and ranks. We have obtained a necessary and sufficient condition for negacyclic codes of odd length over $R$ to be free. In Section IV, we have obtained the complete ideal structure of $\frac{R[x]}{\left \langle x^{2^k}+1 \right \rangle}$. We have classified negacyclic codes of length $2^k$ and their duals over $R$. We have determined the size of each negacyclic code. We have also given the structure of self-orthogonal and self-dual negacyclic codes of length $2^k$ over $R$.

\section{Preliminaries}
Throughout the paper, $R$ denotes the ring $\mathbb{Z}_4+u\mathbb{Z}_4 = \{a + ub ~|~ a, b \in \mathbb{Z}_4\}$ with $u^2=0$. $R$ can be viewed as the quotient ring $\mathbb{Z}_4[u]/\left \langle u^2 \right \rangle$. The units of $R$ are
\[1, 3, 1+u, 1+2u, 1+3u, 3+u, 3+2u, 3+3u~,\]
and the non-units are
\[0,2, u, 2u, 2+u, 2+2u, 3u, 2+3u~.\]
$R$ has six ideals in all, and the lattice diagram of these ideals is as follows:
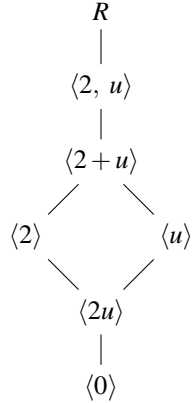
\begin{center}
\begin{tikzpicture}[scale=1]
  \node (1) at (0,3) {$R$};
   \node (2) at (0,2) {{$\langle 2,~u \rangle$}};
    \node (3) at (0,1) {{$\langle 2+u \rangle$}};
     \node (4) at (-1,0) {$\langle 2 \rangle$};
      \node (6) at (1,0) {$\langle u \rangle$};
       \node (8) at (0,-1) {$\langle 2u \rangle$};
        \node (9) at (0,-2) {$\langle 0 \rangle$};
  \draw (3) -- (4) --  (8) --  (6) --  (3);
   \draw (3) -- (2) -- (1);
     \draw (8) -- (9);
\end{tikzpicture}
\end{center}

$R$ is a non-principal local ring of characteristic 4 with $\left \langle 2, u\right \rangle$ as its unique maximal ideal.

The image of an element $f(x) \in R[x]$ in $\overline{R}[x]$ under the projection map $R[x] \rightarrow \overline{R}[x]$ is denoted by $\overline{f}(x)$. A polynomial $f(x) \in R[x]$ is called \emph{basic irreducible (basic primitive)} if $\overline{f}(x)$ is an irreducible (primitive) polynomial in $\overline{R}[x]$.

A linear code $C$ of length $n$ over $R$ is an $R$-submodule of $R^n$. $C$ may not be an $R$-free module. A linear code $C$ of length $n$ over $R$ can be expressed as $C = C_1 + uC_2$, where $C_1, C_2$ are linear codes of length $n$ over $\mathbb{Z}_4$ as $R^n = \mathbb{Z}_4^n + u\mathbb{Z}_4^n$. The Euclidean inner product of any two elements $x=(x_1, x_2, \ldots, x_n)$ and $y=(y_1, y_2, \ldots, y_n)$ of $R^n$ is defined as $x\cdot y = x_1y_1 + x_2y_2 + \cdots + x_ny_n$, where the operation is performed in $R$. The dual of a linear code $C$ is defined as $C^\perp = \{y \in R^n ~|~ x\cdot y = 0 ~\forall x \in C\}$. It follows immediately that if $C = C_1 + uC_2$ is a linear code over $R$, then $C^\perp = C_1^\perp + uC_2^\perp$. We define the \emph{rank} of a code $C$ as the minimum number of generators for $C$ and the \emph{free rank} of $C$ as the rank of $C$ if $C$ is a free module over $R$. There are two other codes associated with $C$, namely Tor$(C)$ and Res$(C)$ and are defined as Tor$(C)=\{ b \in \mathbb{Z}_4^n ~:~ ub \in C \}$ and Res$(C)=\{a \in \mathbb{Z}_4^n ~:~ a+ub \in C ~\mbox{for some} ~b \in \mathbb{Z}_4^n\}$.

A polynomial $f(x)$ over $R$ is called a \textit{regular polynomial} if it is not a zero divisor in $R[x]$, equivalently, $f(x)$ is regular if $\overline{f}(x) \neq 0$.  Two polynomials $f(x), g(x) \in R[x]$ are said to be \emph{coprime} if there exist $a(x), b(x) \in R[x]$ such that $a(x)f(x) + b(x)g(x) = 1,$ or equivalently $\langle f(x) \rangle+\langle g(x) \rangle=R[x]$.  The following result is well known and has been proved in the present setting in \cite{rk}.

\begin{lemma}\cite[Lemma 3.4]{rk}
Let $f(x)$, $g(x) \in R[x]$. Then $f(x)$, $g(x)$ are coprime if and only if their images $\overline{f}(x)$, $\overline{g}(x)$ are coprime in $\overline{R}[x]$.
\end{lemma}

Hensel’s Lemma \cite[Theorem XIII.4]{mcdonald} is an important tool in studying local rings, which can lift the factorization into a product of pairwise coprime polynomials over $\overline{R}$ to such a factorization over $R$.

\begin{theorem}[Hensel's Lemma \cite{pless}]\label{hensels} Let $f$ be a monic polynomial in $R[x]$ and assume that $\overline{f}= g_1 g_2\cdots g_r$, where $g_1,g_2,\ldots,g_r$ are pairwise coprime monic polynomials over $\overline{R}$. Then there exist pairwise coprime monic polynomials $f_1,f_2,\ldots,f_r$ over $R$ such that $f=f_1f_2\cdots f_r$ in $R$ and $\overline{f}_i=g_i$, $i=1,2,\ldots,r$.
\end{theorem}
\begin{theorem}\cite[Theorem 3.1]{rk}\label{primitive} Let $g(x) \in \mathbb{F}_2[x]$ be a monic irreducible (primitive) divisor of $x^{2^r-1}-1$. Then there exists a unique monic basic irreducible (primitive) polynomial $f(x)$ in $R[x]$ such that $\overline{f}(x) = g(x)$ and $f(x)~|~(x^{2^r-1}-1)$ in $R[x]$.
\end{theorem}

We call the polynomial $f(x)$ in Theorem \ref{primitive} the \emph{Hensel lift} of $g(x)$ to $R$.

 Let $n$ be an odd integer. Then it follows from \cite[Theorem XIII.11]{mcdonald} that $x^n+1$ factorizes uniquely into pairwise coprime basic irreducible polynomials over $R$. Let
\[x^n+1= f_1f_2\cdots f_m \]
be such a factorization of $x^n+1$. Then it follows from the Chinese Remainder Theorem that
\[\frac{R[x]}{\left\langle x^n+1\right\rangle} = \oplus_{i=1}^m \frac{R[x]}{\left\langle f_i\right\rangle}~.\]
Therefore every ideal $I$ of $\frac{R[x]}{\left\langle x^n+1\right\rangle}$ can be expressed as $I = \oplus_{i=1}^m I_i$, where $I_i$ is an ideal of the ring $\frac{R[x]}{\left\langle f_i\right\rangle}$, $i=1, 2, \ldots, m$.

Now we consider the Galois extension of $R$. Let $f(x)$ be a basic irreducible polynomial of degree $r$ in $R[x]$. Then the Galois extension of $R$ is defined as the quotient ring $\frac{R[x]}{\left\langle f(x) \right\rangle}$ and is denoted by $GR(R,r)$.

The authors have discussed in \cite{rk} the Galois extension of $R$ and proved that the ring $GR(R,r)$ is a local ring with unique maximal ideal $\left\langle \left\langle 2,~u \right\rangle + \left\langle f \right\rangle \right\rangle$ and the residue field $\mathbb{F}_{2^r}$.

Let $\mathcal{T}=\{0,1, \xi, \xi ^2, \ldots, \xi ^{2^r-2} \}$ be the \textit{Teichm\"{u}ller} representatives of $GR(R,r)$, where $\xi$ is a root of a basic primitive polynomial of degree $r$ in $R[x]$. An element $x$ of $GR(R,r)$ can be represented as $x=a_0+2a_1+ua_2+2ua_3$, where $a_0, a_1,a_2,a_3 \in \mathcal{T}$. A non-zero element $x= a_0 + 2 a_1 + u a_2  + 2u a_3$ of $GR(R,r)$ is unit if and only if $a_0$ is non-zero in $\mathcal{T}$ \cite{rk}.

Thus the group of units of $GR(R,r)$ \cite{rk}, denoted by $GR(R,r)^{*}$, is given by \[GR(R,r)^{*}= \{ a_0 + 2 a_1 + u a_2 + 2u a_3  ~~:~~ a_0, a_1,a_2,a_3 \in \mathcal{T}, a_0 \ne 0 \} .\]

The set of all zero divisors of $GR(R,r)$ is given by $\{ 2a_1+ua_2+2ua_3 ~~ :~~ a_1,a_2,a_3 \in \mathcal{T}   \}$, which is the maximal ideal generated by $\left\langle 2, u \right\rangle$ in $GR(R,r)$.

 The following theorem gives the ideal structure of $GR(R,r)=\frac{R[x]}{\left\langle f \right\rangle}$.

\begin{theorem}\label{ideals of gr}\cite[Theorem 3.5]{rk}
Let $f \in R[x]$ be a basic irreducible polynomial. Then the ideals of $\frac{R[x]}{\left\langle f \right\rangle}$ are precisely, $\{0\}$, $\left\langle 1+\left\langle f\right\rangle\right\rangle$, $\left\langle 2+\left\langle f\right\rangle\right\rangle$, $\left\langle u+\left\langle f\right\rangle\right\rangle$, $\left\langle 2u+\left\langle f\right\rangle\right\rangle$, $\left\langle 2+u+\left\langle f\right\rangle\right\rangle$ and $\left\langle \left\langle 2, u\right\rangle + \left\langle f \right\rangle \right\rangle$.
\end{theorem}

\section{Negacyclic codes of odd lengths over $\mathbb{Z}_4+u\mathbb{Z}_4$}

Let $\lambda$ be a unit in $R$. A linear code $C$ of length $n$ over $R$ is said to be \textit{$\lambda$-constacyclic code} if $C$ is invariant under $\lambda$ cyclic shits,  i. e., if $(c_0,c_1,\ldots,c_{n-1}) \in C$, then $(\lambda c_{n-1}, c_0,c_1, \ldots, c_{n-2}) \in C$. If $\lambda=-1$, then $C$ is called a \textit{negacyclic code}. For $\lambda=1$ negacyclic codes coincide with cyclic codes.

In the polynomial representation of  elements of $R_n$, a $\lambda$-constacyclic code of length $n$ over $R$ is an ideal of $\frac{R[x]}{\left\langle x^n-\lambda \right\rangle}$. In particular for $\lambda=-1$, a negacyclic code is an ideal of $\frac{R[x]}{\left\langle x^n+1 \right\rangle}$.

Let $g(x)=g_0+g_1x+g_2x^2+\cdots+g_rx^r$ be a polynomial in $R_n$. Then the reciprocal of $g(x)$ is the polynomial $g^{*}(x)=x^rg(\frac{1}{x})=g_r+g_{r-1}x+g_{r-2}x^2+\cdots+g_0x^r$ in $R_n$. If $I$ is an ideal of $R_n$, then so is $A(I)^{*}$, where $A(I)^{*}=\{g^{*}(x)~:~g(x) \in A(I) \}$. It is well known that if $C$ is a negacyclic code, then $C^{\perp}=A(C)^{*}$.

\begin{theorem} \label{consta}
If $C=C_1 +u C_2$ is a $\lambda$-constacyclic code of length $n$ over $R$, then $C_1$ is either a cyclic code or a negacyclic code of length $n$ over $\mathbb{Z}_4$.
\end{theorem}

\begin{proof}
Let $T_{\lambda}$ be the $\lambda$-constacylic shift operator on $R^n$. Let $C$ be a $\lambda$-constacyclic code of length $n$ over $R$. Let $(a_0,a_1,\ldots,a_{n-1}) \in C_1$, $(b_0,b_1,\ldots,b_{n-1}) \in C_2$. Then the corresponding element of $C$ is $c=(c_0,c_1,\ldots, c_{n-1})=(a_0,a_1,\ldots,a_{n-1}) +u(b_0,b_1,\ldots,b_{n-1})=(a_0+ub_0,a_1+ub_1,\ldots,a_{n-1}+ub_{n-1})$. Since $C$ is a $\lambda$-constacyclic code,  so $T_{\lambda} (c)=(\lambda c_{n-1},c_0,c_1,\ldots,c_{n-2}) \in C$, where $c_i=a_i+ub_i$. Let $\lambda=\alpha+u\beta$, where $\alpha, \beta \in \mathbb{Z}_4$. Then $T_{\lambda}(c)=(\alpha a_{n-1},a_0,\ldots,a_{n-2})+u((\alpha b_{n-1}+\beta a_{n-1}),b_0,\ldots,b_{n-2})$.  Since the units of $\mathbb{Z}_4$ are $1$ and $-1$, so $\alpha=\pm1$. Hence the result.
\qed \end{proof}

\begin{corollary}
If $C$ is a negacyclic code of length $n$ over $R$, then both $C_1$ and $C_2$ are negacylic codes over $\mathbb{Z}_4$.
\end{corollary}
\begin{proof}
Since $\lambda = -1$, so $\alpha = -1$, $\beta = 0$. Then we get $T_{\lambda}(c)=(-a_{n-1},a_0,\ldots,a_{n-2})+u(-b_{n-1},b_0,\ldots,b_{n-2})$ for some $c \in C$, from which follows that $C_1$, $C_2$ are negacyclic.
\qed \end{proof}

 For the rest of this section, we assume that $n$ is odd. Define $\phi: \frac{R[x]}{\left\langle x^n-1\right\rangle} \rightarrow \frac{R[x]}{\left\langle x^n+1\right\rangle}$ such that $\phi(f(x))=f(-x)$. It was shown in \cite[Theorem 5.1]{dinh} that the map $\phi$ is a ring isomorphism when $R$ is a finite chain ring. The result can easily be extended to arbitrary finite local rings. Therefore $I$ is an ideal of $\frac{R[x]}{\left\langle x^n-1\right\rangle}$ if and only if $J=\phi(I)$ is an ideal of $\frac{R[x]}{\left\langle x^n+1\right\rangle}$.

\begin{theorem} $C$ is a cyclic code of length $n$ over $R$ if and only if $\phi(C)$ is a negacyclic code over $R$.
\end{theorem}
\begin{proof}
Let $\tau$ and $\tau ^\prime$ be cyclic and negacyclic shifts. Then the result follows from the fact that $\phi \circ \tau = \tau ^\prime \circ \phi$.
\qed \end{proof}

The following results (Theorem \ref{notpir} through Theorem \ref{R free to z4 free}) are discussed for cyclic codes over $R$ in \cite{rk}, and are straightforward generalizations thereof via the isomorphism $\phi$ defined above. So we present them here without proofs.

\begin{theorem}\cite[Theorem 4.1]{rk}\label{notpir}
The ring $R_n = \frac{R[x]}{\left\langle x^n+1\right\rangle}$ is not a principal ideal ring.
\end{theorem}

Therefore, a negacyclic code of length $n$ over $R$ is in general not principally generated. Since $n$ is odd, the ring $\frac{\mathbb{Z}_4[x]}{\left\langle x^n+1\right\rangle}$ is a principal ideal ring. Therefore a negacyclic code of length $n$ over $R$ is of the form $C = C_1 + uC_2 = \left\langle g_1\right\rangle + u\left\langle g_2\right\rangle$, where $g_1, g_2 \in \mathbb{Z}_4[x]$ are generator polynomials of the negacyclic codes $C_1$, $C_2$, respectively.

It follows from the Chinese Reminder Theorem that a negacyclic code of length $n$ over $R$ is sum of the ideals listed in Theorem \ref{ideals of gr}.

\begin{theorem}\cite[Corlloary 1]{rk}
The number of negacyclic codes of length $n$ over $R$ is $7^m$, where $m$ is the number of distinct basic irreducible factors of $x^n+1$.
\end{theorem}

The following result gives a sufficient condition for a negacyclic code $C$ over $R$ to be a free $\mathbb{Z}_4$-code.

\begin{theorem}\label{freeZ4}\cite[Theorem 4.4]{rk}
Let $C = C_1 + uC_2$ be a negacyclic code of length $n$ over $R$. If $C_1$, $C_2$ are free codes over $\mathbb{Z}_4$, then $C$ is  a free $\mathbb{Z}_4$-module.
\end{theorem}

The converse of the above theorem is in general not true, i.e., if a negacyclic code $C = C_1 + uC_2$ is a free $\mathbb{Z}_4$-module of length $n$ over $R$, then $C_1$ or $C_2$ may not be a free code of length $n$ over $\mathbb{Z}_4$ (see Example \ref{ex2}). However, if $C$ is an $R$-free module (code) of length $n$ over $R$, then $C_1$ must be a  free code over $\mathbb{Z}_4$ (see Theorem \ref{R free to z4 free}).

\begin{example}
The polynomial $x^{15}-1$ factorizes into irreducible polynomials over $\mathbb{F}_2$ as $x^{15}-1 = (x-1)(x^4+x^3+1)(x^4+x+1)(x^4+x^3+x^2+x+1)(x^2+x+1)$. The Hensel lifts of $x^4+x^3+1$, $x^4+x+1$, $x^4+x^3+x^2+x+1$ and $x^2+x+1$ to $\mathbb{Z}_4 $ are $x^4-x^3+2x^2+1$, $x^4+2x^2-x+1$, $x^4+x^3+x^2+x+1$ and $x^2+x+1$, respectively. Therefore $x^{15}-1=(x-1)(x^4-x^3+2x^2+1)(x^4+2x^2-x+1)(x^4+x^3+x^2+x+1)(x^2+x+1)$. Replacing $x$ by $-x$, we get  $x^{15}+1=(x+1)(x^4+x^3+2x^2+1)(x^4+2x^2+x+1)(x^4-x^3+x^2-x+1)(x^2-x+1)$. Define $C = \left\langle x^4-x^3+x^2-x+1\right\rangle + u \left\langle x^4+2x^2+x+1\right\rangle$. Then $C$ is a negacyclic code of length $15$ over $R$, which is also a free $\mathbb{Z}_4$-module.
\end{example}

\begin{example}\label{ex2} Let $C=C_1+uC_2$ be a free $\mathbb{Z}_4$-negacyclic code of length $7$ over $R$ generated by $g(x)=2x^2+u(x^3+x+1)$. Then $C_1$ is a  negacyclic code of length $7$ over $\mathbb{Z}_4$ generated by $g(x)~ (\mbox{mod}~ u)=g_1(x)=2x^2$. Since $g_1(x)$ is a zero divisor, so $C_1$ is not $\mathbb{Z}_4$-free.
\end{example}

The general form of a cylic code $\mathcal{C}$ of length $n$ over $R$ is given by $\mathcal{C}=\left\langle g(x)+up(x), ~ua(x) \right\rangle$, where $g(x)$, $p(x)$, $a(x) \in \mathbb{Z}_4[x]$ \cite{rk}. The same structure can be adopted to the negacyclic codes of length $n$ over $R$ through the isomorphism $\phi$ defined earlier in this section. Thus a negacyclic code $C$ of length $n$ over $R$ can be expressed as $\mathcal{C}=\left\langle g(x)+up(x), ~ua(x) \right\rangle$, where $g(x)$, $p(x)$, $a(x) \in \mathbb{Z}_4[x]$.

Theorem \ref{thm4.9} and Theorem \ref{thm4.10} below give the minimal spanning set of negacyclic code over $R$.

\begin{theorem}\label{thm4.9} \cite[Theorem 4.8]{rk} Let $C$ be a negacyclic code of length $n$ over $R$. If $C=\left\langle g(x)+up(x), ~ua(x) \right\rangle$, $deg~g(x)=k_1$ and $deg~a(x)=k_2$, then $C$ has rank $2n-k_1-k_2$ and a minimal spanning set $A=\{ (g(x)+up(x)), x(g(x)+up(x)), x^2(g(x)+up(x)), \cdots, x^{n-k_1-1}(g(x)+up(x)),~ua(x), xua(x),x^2ua(x), \cdots$, $x^{n-k_2-1}ua(x)\}$.
\end{theorem}

\begin{theorem} \label{thm4.10} \cite[Theorem 4.10]{rk} Let $C=\left\langle g(x)+up(x), ~ua(x) \right\rangle$ be a negacyclic code of length $n$ over $R$, and $g(x)$ is regular and $a(x)$ is monic with $deg~g(x)=k_1$ and $deg~a(x)=k_2$, respectively. Then $C$ has rank $n-k_2$ and a minimal spanning set $B=\{ (g(x)+up(x)),$ $x(g(x)+up(x)),$ $x^2(g(x)+up(x)), \cdots, $ $x^{n-k_1-1}(g(x)+up(x)),$ $~ua(x),$ $xua(x), x^2ua(x),$ $\cdots, x^{k_1-k_2-1}ua(x) \}$.
\end{theorem}

\begin{example}
Consider the negacyclic code $C$ of length $7$ over $R$ generated by the polynomial $g(x) = x^3+(2+u)x^2+(1+u)x+(1+u)$ and $a(x)=x+1$. Then the rank of $C$ is $6$ and a minimal spanning set of $C$ is $\{ g(x), xg(x), x^2g(x), x^{3}g(x),~ua(x),uxa(x)\}$.
\end{example}

In Theorem \ref{thm4.9} and Theorem \ref{thm4.10}, we have only seen minimal spanning sets of negacyclic codes over $R$. However we can find a basis for negacyclic code of length $n$ when it is principally generated.  The following theorem gives a necessary and sufficient condition for the negacyclic codes over $R$ to be free.

\begin{theorem}\cite[Proposition 1]{rk} Let $C$ be a principally generated negacyclic code of length $n$ over $R$. Then $C$ is free if and only if there exists a monic generator $g(x)$ in $C$ such that $g(x) ~| ~x^n+1$. Furthermore, $C$ has free rank $n- deg~g(x)$ and the elements $g(x)$, $xg(x)$, $\cdots$, $x^{n-deg~g(x)-1}g(x)$ form a basis for $C$.
\end{theorem}

\begin{example}
Consider the negacyclic code $C$ of length $15$ over $R$ generated by the polynomial $g(x) = x^4+2x^2+x+1$, where $g(x)$ is the Hensel lift of $x^4+x+1 \in \mathbb{F}_2[x]$ to $R$ and $g(x) \mid x^{15}+1$. The negacyclic code  $C = \left\langle g(x) \right\rangle$ is an $R$-free negacyclic code of length $15$ and the free rank $11$.
\end{example}

\begin{theorem} \label{R free to z4 free}\cite[Theorem 5.4]{rk} If $C=C_1+uC_2$ is a free negacyclic code of length $n$ over $R$, then so is $C_1$ over $\mathbb{Z}_4$.
\end{theorem}

\begin{example}
 Consider again the negacyclic code $C$ of length $15$ generated by $g(x)=x^4+2x^2+x+1$. Then $C$ is free over $R$, as $x^4+2x^2+x+1$ is a divisor of $x^{15}+1$ over $R$.  Since $x^4+2x^2+x+1$ is a divisor of $x^{15}+1$ over $\mathbb{Z}_4$ as well, so $C_1$ is a free negacyclic code of length $15$ over $\mathbb{Z}_4$.
\end{example}

\section{Negacyclic codes of length $2^k$}

So far we considered negacyclic codes with the assumption that the code length $n$ is coprime to the characteristic of the ring $R$, i.e, $(n,4)=1$ . Now we extend our study to negacyclic codes of length $n=2^k$, $k\geq1$. $(n,4) \neq 1$. Negacyclic codes whose lengths are not relatively prime to characteristic of $R$ are known as repeated root negacyclic codes.

\begin{lemma}\cite{yildiz1} Every non-zero, non-unit element of $R_n=\frac{R[x]}{\left \langle x^n+1\right \rangle}$ must be a zero divisor. \end{lemma}
The following result is a generalization of \cite[Theorem 6.2]{rk}.
\begin{theorem}
The ring $R_n$ is not a local ring when $n=2^{k}m$, where $m=4d+1~ (d\geq1)$ or $m=4d+3~ (d\geq0)$.
\end{theorem}

\begin{theorem} \label{local}
The ring $R_n$ is a local ring when $n=2^{k}$ where $k \geq 1$.
\end{theorem}
\begin{proof}
Define the map $\Phi ~:~ R_n \rightarrow \frac{\mathbb{Z}_4[x]}{\left \langle x^{n}+1 \right \rangle} $ such that $\Phi(f(x))=f_1(x)~~ (\mbox{mod} ~u)$, where $f(x)=f_1(x)+uf_2(x)$. It is easy to verify that $\Phi$ is a surjective ring homomorphism. It is well known that the ring  $\frac{\mathbb{Z}_4[x]}{\left \langle x^{2^k}+1 \right \rangle}$ is a local ring with the unique maximal ideal $\left \langle x+1 \right \rangle$ \cite{dinh1}. The inverse image of the maximal ideal $\left \langle x+1 \right \rangle$ is $\Phi^{-1}(\left \langle x+1 \right \rangle) = \left \langle u,  x+1 \right \rangle$. Since  $\langle u,  x+1 \rangle$ contains all non-units of $R_n$, so it is unique maximal ideal of $R_n$ and hence $R_n$ is local.
\qed \end{proof}

Now onward $n=2^k$, $k \geq 1$. Also, now onward, we prefer to express a polynomial in terms of $x+1$, rather than in $x$, to make the computations easier in $R_n$. Each polynomial in $R_n$ can uniquely be written as $\sum\limits\limits_{j=0}^{n-1} f_j (x+1)^j$, $f_j \in R$, and such a polynomial is denoted by $f(x)$.

Let $R_n^\prime=\frac{\mathbb{Z}_4[x]}{\left \langle x^n+1\right \rangle}$ and $R_n^{\prime \prime}=\frac{\mathbb{Z}_2[x]}{\left \langle x^n+1\right \rangle}$.

\begin{lemma} \label{nilpo} In $R_n$, $(x+1)^n=2x^{\frac{n}{2}}$ and $(x+1)$ is nilpotent with nilpotency $2n$. \end{lemma}
\begin{proof} It can easily be seen by induction that $(x+1)^n=x^n+1+2x^{\frac{n}{2}}$ in $R_n$. Since $x^n=-1$ in $R_n$, so $(x+1)^n=2x^{\frac{n}{2}}$. From this follows that $(x+1)^{2n}=0$. Also there is no $l < 2n$ such that $(x+1)^l=0$ in $R_n$. For if, there any $n < l < 2n$ such that $(x+1)^l=0$, then we get $2(x+1)^{l-n}=0$, since $(x+1)^l=(x+1)^n(x+1)^{l-n}=2x^{\frac{n}{2}}(x+1)^{l-n}$ and $x^{\frac{n}{2}}$ is a unit in $R_n$. But $2(x+1)^{l^\prime} \neq 0$ for any $l^\prime < n$ in $R_n$.  Hence $2n$ is the nilpotency of $x+1$. \qed \end{proof}

\begin{lemma} \label{unit} An element $f(x)=\sum\limits_{j=0}^{n-1} a_j (x+1)^j$ is a unit in $R_n$ if and only if $a_0$ is a unit in $R$. \end{lemma}

\begin{lemma} \label{2} In $R_n$, $(x+1)^n=2x^{\frac{n}{2}}=2\left((x+1)^{\frac{n}{2}}+1\right)$ and also  $\left \langle (x+1)^n \right \rangle = \left \langle 2 \right \rangle$. \end{lemma}

 An element $f(x)$ in $R_n$ can be written as $f(x)=f_1(x)+uf_2(x)$, where $f_i(x) \in R_n^\prime$. An element $f(x)$ in $R_n$ is called a \emph{monic element} if $f(x)$ is a monic polynomial in $R_n$. Now we consider the ideal structure of $R_n$.

Define $\psi: R \rightarrow \mathbb{Z}_4$ such that $\psi(a+bu)=a~ (\mbox{mod}~ u)$. It can easily be seen that $\psi$ is a ring homomorphism with ker $\psi$ $= \left\langle u \right \rangle$ $=u\mathbb{Z}_4$.  Extend $\psi$ to the homomorphism $\Phi: R_n \rightarrow R_n^\prime$ such that $\Phi(a(x)+ub(x))=a(x)~ (\mbox{mod}~ u)$, where $a(x), b(x) \in \mathbb{Z}_4[x]$. Let $I$ be an ideal of  $R_n$. Restrict $\Phi$ to $I$ and define
\[J=\{ h(x) \in R_n^\prime ~:~ uh(x) \in \mbox{ker}~\Phi\}~.\]
 Clearly $J$ is an ideal of $R_n^\prime$. Since $R_n^\prime$ is  finite chain ring with the maximal ideal $\left \langle x+1 \right \rangle$, so $J=\left \langle (x+1)^m \right \rangle$ for some $1 \leq m \leq 2n$. Therefore ker $\Phi$ $= \left \langle u(x+1)^m \right\rangle$. Similarly, the image of $I$ under $\Phi$ is an ideal of $R_n^\prime$ and $\Phi(I)=\left \langle (x+1)^s \right \rangle$ for some $1 \leq s \leq 2n$. Hence $I=\left \langle (x+1)^s+up(x),~u(x+1)^m \right \rangle$ for some $p(x)=\sum\limits_{j=0}^{n-1}p_j(x+1)^j \in \mathbb{Z}_4[x]$. Since $u(x+1)^s= u((x+1)^s+up(x)) \in C$ and $\Phi(u(x+1)^s)=0$, so $(x+1)^m \mid (x+1)^s$. This implies that $m \leq s$. When $m=s$,  we get $u(x+1)^m \in \left  \langle (x+1)^s+up(x) \right \rangle=I$. Now let $m < s$. Then a  non-trivial ideal $I$ of $R_n$  has the form
 \[I = \left \langle (x+1)^s+u\sum\limits_{j=0}^{n-1}p_j(x+1)^j,~u(x+1)^m \right \rangle~, ~~1 \leq s \leq 2n-1~~\mbox{and}~~ 0 \leq m \leq s-1.\]

 When $m < n$, \[I = \left \langle (x+1)^s+u\sum\limits_{j=0}^{n-1}p_j(x+1)^j,~u(x+1)^m \right \rangle=\left \langle (x+1)^s+u\sum\limits_{j=0}^{m-1}p_j(x+1)^j,~u(x+1)^m \right \rangle.\]
 Therefore \[I= \left \langle (x+1)^s+u\sum\limits_{j=0}^{\mbox{min}\{m-1,n-1\}}p_j(x+1)^j,~u(x+1)^m \right \rangle.\]

 If $t$ is the smallest non-zero integer such that $p_t$ is non-zero, then a polynomial  $f(x)=  (x+1)^s+u \sum\limits_{j=0}^{n-1} p_j(x+1)^j \in R[x]$ can be represented as $f(x)= (x+1)^s+u(x+1)^t h(x)$, where $h(x) \in \mathbb{Z}_4[x]$ and deg $h(x)\leq n-t-1$. Hence $I$ can be written as \[I= \left \langle (x+1)^s+u(x+1)^th(x),~u(x+1)^m \right \rangle,\] where $1 \leq s \leq 2n-1$, $0 \leq t <$ min$\{m,n\}$, $0 \leq m \leq s-1$ and $h(x) \in \mathbb{Z}_4[x]$.

 Summarizing this discussion, we present the complete ideal structure of $R_n$ in the following theorem.

\begin{theorem}\label{ideals} Let $I$ be an ideal in $R_n$. Then $I$ is one of the following:
\begin{enumerate}
\item Trivial ideals:

$\left \langle  0 \right \rangle$ or $\left \langle  1 \right \rangle$.
\item Principal ideals:
\begin{enumerate}
\item[(a)]  $\left \langle  u(x+1)^{m}\right \rangle$, $0 \leq m \leq 2n-1$
\item[(b)] $\left \langle  (x+1)^s+u(x+1)^th(x)\right \rangle$, $1 \leq s \leq 2n-1$, $0\leq t <$ min{$\{s,n\}$}.
\end{enumerate}
\item Non-principal ideals:

$\left \langle  (x+1)^s+u(x+1)^th(x),~u(x+1)^m \right \rangle$,  $1 \leq s \leq 2n-1$, $0\leq t <$ min{$\{m,n\}$}, $0 \leq m \leq s-1$.
\end{enumerate}
\end{theorem}

The ideals described in  Theorem \ref{ideals} are not distinct. For instance, in $R_2$, the ideals  $\left \langle  (x+1)^3+u \right \rangle$  and $\left \langle  (x+1)^3+u(1+(x+1)) \right \rangle$ are same, since  $(x+1)^3+u(1+(x+1))=((x+1)^3+u)(1+(x+1))$, as $(1+(x+1))$ is a unit in $R_n$. Similarly, the ideals $\left \langle  (x+1)^2+u \right \rangle $ and  $ \left \langle  (x+1)^2+3u) \right \rangle$ are also same. But the ideals $\left \langle  (x+1)^2+u \right \rangle $ and $\left \langle  (x+1)^2+u(1+(x+1)) \right \rangle$ are distinct, as $u(x+1)$ is neither in $\left \langle  (x+1)^2+u \right \rangle $ nor $\left \langle  (x+1)^2+u(1+(x+1)) \right \rangle$. So it is very important to find the smallest value of $T$ such that $u(x+1)^T \in \left \langle  (x+1)^s+u(x+1)^th(x) \right \rangle$, through which the repetition of ideals can be avoided and ideals (negacyclic codes) can be determined distinctly.

\begin{theorem}\label{minvalue} Let $T$ be the smallest non-negative integer such that \[u(x+1)^T \in I=\left \langle  (x+1)^s+u(x+1)^th(x) \right \rangle,\] where $1 \leq s \leq n-1$, $0 \leq t < s$, $h(x) \in \mathbb{Z}_4[x]$ and deg $h(x) \leq s-t-1$. Then $T=s$.
\end{theorem}

\begin{theorem}\label{minvalue1} Let $I=\left \langle  (x+1)^s+u(x+1)^t h(x) \right \rangle$, where $n \leq s \leq 2n-1$ and deg $h(x) \leq n-t-1$. Then there exists a non-negative integer $T < n$ such that $u(x+1)^T \in I$ if and only if $h(x)$ is a unit in $R_n^\prime$ and $t < s-n$. Moreover, if $T$ is the smallest of such non-negative integers, then $T=2n-s+t $.
\end{theorem}

\begin{theorem}\label{minvalue2} Let $T$ be the smallest non-negative integer such that $u(x+1)^T  \in \left \langle  (x+1)^s+u(x+1)^t h(x) \right \rangle$, where $n \leq s \leq 2n-1$, $t \geq s-n$, $h(x)$ is a unit in $R_n^\prime$ and deg $h(x) \leq n-t-1$. Then  $T \geq n$ and $T=\mbox{min}\{ s, 2n-s+t \} $.
\end{theorem}

  Theorem \ref{minvalue1} and  Theorem \ref{minvalue2} give the value of $T$ such that $u(x+1)^T \in I$, only when $h(x)$ is a unit in $\frac{\mathbb{Z}_4[x]}{\langle x^n+1 \rangle}$. If $h(x)$ is not a unit in $\frac{\mathbb{Z}_4[x]}{\langle x^n+1 \rangle}$, then we can have either $h(x)=2h^\prime(x)$ or $h(x)=2h_1(x)+(x+1)^lh_2(x)$, where $h^\prime(x)$, $h_1(x)$ and $h_2(x)$ are units in $R_n^{\prime \prime}$, $R_n^\prime $, respectively. For example,
  \begin{eqnarray*}
  h(x)&=&2+2(x+1)+(x+1)^3+3(x+1)^4+2(x+1)^5+(x+1)^6+3(x+1)^7 \\
  &=&2(1+(x+1)+(x+1)^5)+(x+1)^3(1+3(x+1)+(x+1)^3+3(x+1)^3) \\
  &=&2h_1(x)+(x+1)^3h_2(x),
  \end{eqnarray*}
  where $h_1(x)$, $h_2(x)$ are units in $R_n^{\prime \prime}$, $R_n^\prime $, respectively.

  Now we will find the smallest value of $T$ in these two cases of $h(x)$ also.

\begin{theorem} \label{minvalue3} Let $T$ be the smallest non-negative integer such that \[u(x+1)^T \in \left \langle (x+1)^s+2u(x+1)^t h(x) \right \rangle,\] where $n \leq s \leq 2n-1$, $0 \leq t < n$ and $h(x)$ is a unit in $R_n^{\prime \prime}$. Then $t < s-n$ and $T=\mbox{min}\{s,3n-s+t\}$.
\end{theorem}

\begin{theorem} \label{minvalue4}Let $T$ be the smallest non-negative integer such that \[u(x+1)^T \in \left \langle (x+1)^s+u(x+1)^t (2h_1(x)+(x+1)^lh_2(x)) \right \rangle,\] where $n \leq s \leq 2n-1$, $0 \leq t < n$, $h_1(x)$ and $h_2(x)$ are units in $R_n^{\prime \prime}$ . Then $T=\mbox{min}\{s,2n-s+t+l\}$.
\end{theorem}
\begin{proof} The proof is similar to that of Theorem \ref{minvalue3}.
\qed \end{proof}

We summarize the value of $T$ for different cases of $h(x)$, $s$ and $t$ in the following remarks.

\begin{remark}\label{final T} In view of Theorems \ref{minvalue} through \ref{minvalue4},  if $0 \leq T \leq 2n-1$ is the smallest non-negative integer such that $u(x+1)^T \in I=\left \langle  (x+1)^s+u (x+1)^th(x) \right \rangle,$ where $0 \leq s \leq 2n-1$ and deg $h(x) \leq n-t-1$, then \[ T=\begin{cases}
s & ~\mbox{if}~1 \leq s \leq n-1, \\
2n-s+t & ~\mbox{if}~n \leq s \leq 2n-1, ~0 \leq t < 2s-2n~~\mbox{and}~~ h(x)~\mbox{is a unit in}~ R_n^\prime,\\
s & ~\mbox{if}~n \leq s \leq 2n-1~~\mbox{and}~t \geq 2s-2n, ~~\mbox{and}~~ h(x)~\mbox{is a unit in} ~R_n^\prime,\\
3n-s+t & ~\mbox{if}~n < s \leq 2n-1,~0 \leq t < 2s-3n~~\mbox{and}~h(x)=2h^\prime(x),\\
s & ~\mbox{if}~n < s \leq 2n-1,~ 2s-3n \leq t < s-n~~\mbox{and}~h(x)=2h^\prime(x),\\
2n-s+l+t & ~\mbox{if}~n < s \leq 2n-1,~0 \leq t < 2s-2n,~0 \leq l+t < 2s-2n~~\mbox{and}~h(x)=2h_1(x)+(x+1)^lh_2(x),\\
s & ~\mbox{if}~n < s \leq 2n-1,~0 \leq t < s-n,~ l+t \geq 2s-2n~~\mbox{and}~h(x)=2h_1(x)+(x+1)^lh_2(x).\\
\end{cases}\]
\end{remark}

\begin{remark}\label{final T_1} If $0 \leq T_1 \leq n-1$ is the smallest non-negative integer such that $2u(x+1)^{T_1} \in I=\left \langle  (x+1)^s+u (x+1)^th(x) \right \rangle,$ where $0 \leq s \leq 2n-1$ and deg $h(x) \leq n-t-1$, then \[ T_1=\begin{cases}
0 & ~\mbox{if}~1 \leq s \leq n-1, \\
0 & ~\mbox{if}~n \leq s \leq 2n-1, ~0 \leq t < s-n~~\mbox{and}~~ h(x)~\mbox{is a unit in} ~R_n^\prime,\\
n-s+t & ~\mbox{if}~n \leq s \leq 2n-1, ~s-n \leq t < 2s-2n~~\mbox{and}~~ h(x)~\mbox{is a unit in}~ R_n^\prime,\\
s-n & ~\mbox{if}~n \leq s \leq 2n-1, ~2s-2n \leq t < n~~\mbox{and}~~ h(x)~\mbox{is a unit in}~ R_n^\prime,\\
2n-s+t & ~\mbox{if}~n < s \leq 2n-1,~0 \leq t < 2s-3n~~\mbox{and}~h(x)=2h^\prime(x),\\
s-n & ~\mbox{if}~n \leq s \leq 2n-1,~ 2s-3n \leq t < s-n~~\mbox{and}~h(x)=2h^\prime(x),\\
 0 & ~\mbox{if}~n \leq s \leq 2n-1,~0 \leq t < s-n,~0 \leq l+t \leq s-n~~\mbox{and}~h(x)=2h_1(x)+(x+1)^lh_2(x),\\
l+t-s+n & ~\mbox{if}~n < s \leq 2n-1,~0 \leq t < s-n,~s-n < l+t < 2s-2n~~\mbox{and}~h(x)=2h_1(x)+(x+1)^lh_2(x),\\
s-n & ~\mbox{if}~n < s \leq 2n-1,~0 \leq t < s-n,~ l+t \geq 2s-2n~~\mbox{and}~h(x)=2h_1(x)+(x+1)^lh_2(x).\\
\end{cases}\]
\end{remark}

Now we can distinguish the ideals of $R_n$. The following the Theorem gives the distinct monic principal ideals of $R_n$.

\begin{theorem}\label{dpi}
The distinct monic principal ideals of $R_n$ are
 \begin{enumerate}
  \item $I=\left \langle  (x+1)^s+u (x+1)^t h(x) \right \rangle$, where $1 \leq s \leq n-1$, $0 \leq t < s$, $h(x)$ is either zero or a unit in $R_n^{\prime \prime}$ and deg $h(x) \leq s-t-1$.
\item $I=\left \langle  (x+1)^s+u (x+1)^t h(x) \right \rangle$, where $n \leq s \leq 2n-1$, $0 \leq t < n$, $h(x)$ is either zero or a unit in $R_n^{\prime \prime}$ and deg $h(x) \leq T-t-1 = \begin{cases}
    2n-s-1 & ~\mbox{if} ~ t < s-n, \\
    n-t-1 & ~\mbox{if} ~ t \geq s-n.
    \end{cases}$
    \item $I=\left \langle  (x+1)^s+2u (x+1)^t h(x) \right \rangle$, where $n < s \leq 2n-1$, $0 \leq t < s-n$, $h(x)$ is a unit in $R_n^{\prime \prime}$ and \\deg $h(x) \leq T_1 -t -1 = \begin{cases}
    2n-s-1 & ~\mbox{if} ~ t < 2s-3n, \\
    s-n-t-1 & ~\mbox{if} ~ t \geq 2s-3n.
    \end{cases}$
    \item $I=\left \langle  (x+1)^s+u (x+1)^t (2h_1(x)+(x+1)^lh_2(x)) \right \rangle$, where $n < s \leq 2n-1$, $0 \leq t < s-n$, $l > s-n-t$, $h_1(x)$, $h_2(x)$ are units in $R_n^{\prime \prime}$, deg $h_1(x) \leq T_1-t-1$ and deg $h_2(x) \leq n-t-l-1$.
 \end{enumerate}
 \end{theorem}

Using the above principal ideals of $R_n$, the non-principal ideals can be described as follows:

\begin{theorem}
The distinct non-principal ideals of $R_n$ are
 \begin{enumerate}
 \item  $I=\left \langle  (x+1)^s+u (x+1)^t h(x), ~u(x+1)^m \right \rangle$, \\where $1 \leq s \leq n-1$, $0 \leq t < m < s$, $ h(x)$ is either zero or a unit in $R_n^{\prime \prime}$ and deg $h(x) \leq m-t-1$.
\item $I=\left \langle  (x+1)^s+u (x+1)^t h(x), ~u(x+1)^m \right \rangle$, \\where $n \leq s \leq 2n-1$, $0 \leq t < n$, $1+t \leq m < T$, $ h(x)$ is either zero or a unit in $R_n^{\prime \prime}$ and deg $h(x) \leq ~\mbox{min}\{m,n\}-t-1$.
\item  $I=\left \langle  (x+1)^s+2u (x+1)^t h(x), ~u(x+1)^m \right \rangle$, \\where $n < s \leq 2n-1$, $0 \leq t < s-n$, $1+t \leq m < T$, $h(x)$ is a unit in $R_n^{\prime \prime}$ and deg $h(x) \leq ~\mbox{min}\{m,T_1\}-t-1$.
\item   $\left \langle  (x+1)^s+2u (x+1)^t h(x), ~2u(x+1)^{m_1} \right \rangle$, \\where $n < s \leq 2n-1$, $0 \leq t < s-n$, $1+t \leq m_1 < T_1$, $h(x)$ is a unit in $R_n^{\prime \prime}$ and deg $h(x) \leq m_1-t-1$.
\item  $I=\left \langle  (x+1)^s+u (x+1)^t (2h_1(x)+(x+1)^lh_2(x)), ~u(x+1)^m \right \rangle$, \\ where $n < s \leq 2n-1$, $0 \leq t < s-n$, $t+l < m  < n$, $h_1(x)$, $h_2(x)$ are units in $R_n^{\prime \prime}$, deg $h_1(x) < T_1$, deg $h_2(x) < n-t-l$.
\item $\left \langle  (x+1)^s+u (x+1)^t (2h_1(x)+(x+1)^lh_2(x)), ~2u(x+1)^m_1 \right \rangle$, \\ where $n < s \leq 2n-1$, $0 \leq t < s-n$, deg $h_1(x) < m_1  < ~\mbox{min}\{s-n,n-s+l+t\}$, $h_1(x)$, $h_2(x) $ are units in $R_n^{\prime \prime}$.
 \end{enumerate}
\end{theorem}

\begin{theorem}\label{negacodes} Let $C$ be a negacyclic code of length $n$ over $R$. Then $C$ is one of the following:
\begin{itemize}
\item \textbf{Type 0:}  $\left \langle  0 \right \rangle$ or $\left \langle  1 \right \rangle$.
\item \textbf{Type 1:} $\left \langle  u(x+1)^{m}\right \rangle$, $0 \leq m \leq 2n-1$.
\item \textbf{Type 2.0:} $\left \langle  (x+1)^s+u (x+1)^t h(x) \right \rangle$, where $1 \leq s \leq n-1$, $0 \leq t < s$,\\ $ h(x)$ is either zero or a unit in $R_n^{\prime \prime}$ and deg $h(x) \leq s-t-1$.
\item \textbf{Type 2.1:} $\left \langle  (x+1)^s+u (x+1)^t h(x) \right \rangle$, where $n \leq s \leq 2n-1$, $0 \leq t < n$,\\ $ h(x)$ is either zero or a unit in $R_n^{\prime \prime}$ and deg $h(x) \leq T-t-1 =\begin{cases}
    2n-s-1 & ~\mbox{if} ~ t < s-n, \\
    n-t-1 & ~\mbox{if} ~ t \geq s-n.
    \end{cases}$
\item \textbf{Type 2.2:} $\left \langle  (x+1)^s+2u (x+1)^t h(x) \right \rangle$, where $n < s \leq 2n-1$, $0 \leq t < s-n$,\\ $h(x)$ is a unit in $R_n^{\prime \prime}$ and deg $h(x) \leq T_1-t-1 = \begin{cases}
    2n-s-1 & ~\mbox{if} ~ t < 2s-3n, \\
    s-n-t-1 & ~\mbox{if} ~ t \geq 2s-3n.
    \end{cases}$
\item  \textbf{Type 2.3:} $\left \langle  (x+1)^s+u (x+1)^t (2h_1(x)+(x+1)^lh_2(x)) \right \rangle$, where $n < s \leq 2n-1$, $0 \leq t < s-n$, \\ $l >s-n-t$, $h_1(x)$, $h_2(x)$ are units in $R_n^{\prime \prime}$, deg $h_1(x) \leq T_1-t-1$ and \\deg $h_2(x) \leq n-t-l-1$.
    \item  \textbf{Type 3.0:} $\left \langle  (x+1)^s+u (x+1)^t h(x), ~u(x+1)^m \right \rangle$, where $1 \leq s \leq n-1$, $0 \leq t < m < s$,\\ $ h(x)$ is either zero or a unit in $R_n^{\prime \prime}$ and deg $h(x) \leq m-t-1$.
        \item  \textbf{Type 3.1:} $\left \langle  (x+1)^s+u (x+1)^t h(x), ~u(x+1)^m \right \rangle$, where $n \leq s \leq 2n-1$, $0 \leq t < n$, $1+t \leq m < T$,\\ $ h(x)$ is either zero or a unit in $R_n^{\prime \prime}$ and deg $h(x) \leq m-t-1$.
\item  \textbf{Type 3.2:} $\left \langle  (x+1)^s+2u (x+1)^t h(x), ~u(x+1)^m \right \rangle$, where $n < s \leq 2n-1$, $0 \leq t < s-n$, \\$1+t \leq m < n$, $h(x)$ is a unit in $R_n^{\prime \prime}$ and deg $h(x) \leq ~\mbox{min}\{m,T_1\}-t-1$.
\item  \textbf{Type 3.3:} $\left \langle  (x+1)^s+2u (x+1)^t h(x), ~2u(x+1)^{m_1} \right \rangle$, where $n \leq s \leq 2n-1$, \\ $0 \leq t < s-n$, $1+t \leq m_1 < T_1$, $h(x)$ is a unit in $R_n^{\prime \prime}$ and deg $h(x) \leq m_1-t-1$.
\item \textbf{Type 3.4:} $\left \langle  (x+1)^s+u (x+1)^t (2h_1(x)+(x+1)^lh_2(x)), ~u(x+1)^m \right \rangle$, where $n < s \leq 2n-1$,\\ $0 \leq t < s-n$, $t+l < m  < n$,  $h_1(x)$, $h_2(x) $ are units in $R_n^{\prime \prime}$.
\item \textbf{Type 3.5:} $\left \langle  (x+1)^s+u (x+1)^t (2h_1(x)+(x+1)^lh_2(x)), ~2u(x+1)^{m_1} \right \rangle$, where $n < s \leq 2n-1$, \\$0 \leq t < s-n$, $1+t < m_1  < T_1$, $h_1(x)$, $h_2(x)$ are units in $R_n^{\prime \prime}$.
 \end{itemize}
 \end{theorem}

In the following Theorems \ref{duals1} and \ref{duals2}, we give the annihilators of each ideal of types described in Theorem \ref{negacodes}.

\begin{theorem}\label{duals1} Let $I$ be a non-trivial principal ideal and $A(I)$ be its annihilator in $R_n$.
\begin{enumerate}
\item[(a)] If $I$= $\left \langle  u(x+1)^m \right \rangle$, where $0 \leq m \leq 2n-1$, then $A(I)=\left \langle  (x+1)^{2n-m},~u \right \rangle$.

 \item[(b)] If $I=\left \langle  (x+1)^s+u (x+1)^t h(x) \right \rangle$, where $1 \leq s \leq n-1$, $0 \leq t < s$, $ h(x)$ is either zero or a unit in $R_n^{\prime \prime}$ and deg $h(x) \leq s-t-1$, then
 \[A(I)=\begin{cases}  \left \langle  (x+1)^{2n-s}+ u(x+1)^{2n-2s+t} h(x) \right \rangle, & ~\mbox{when}~ t < 2s-n,\\
 \left \langle  (x+1)^{2n-s}+2u(x+1)^{n-2s+t}(1+(1+x)^{\frac{n}{2}})h(x) \right \rangle,  & ~ \mbox{when}~ t \geq 2s-n.
 \end{cases}\]

 \item[(c)] If $I=\left \langle  (x+1)^{s}+u(x+1)^th(x) \right \rangle$, where $n \leq s \leq 2n-1$, $0 \leq t \leq n-1$ and $ h(x)$ is either zero or a unit in $R_n^{\prime \prime}$, then
 \[A(I)=\begin{cases}
 \left \langle (x+1)^{s-t}+uh(x), ~u(x+1)^{2n-s} \right \rangle, & ~\mbox{when}~t < 2s-2n, \\
 \left \langle (x+1)^{2n-s}+u(x+1)^{2n-2s+t}h(x) \right \rangle,  & ~ \mbox{when} ~ t \geq 2s-2n.
 \end{cases}\]

 \item[(d)] If $I=\left \langle  (x+1)^{s}+2u(x+1)^th(x) \right \rangle $, where $n < s \leq  2n-1$ and $0 \leq t < s-n$ and $h(x)$ is a unit in $R_n^{\prime \prime}$, then
 \[A(I)=\begin{cases}
 \left \langle  (x+1)^{2n-s}+u(x+1)^{3n-2s+t}(1+(x+1)^{\frac{n}{2}})h(x) \right \rangle, &~\mbox{when}~ 2s-3n \leq t < s-n,\\
 \left \langle  (x+1)^{n-t}, ~u(x+1)^{2n-s}  \right \rangle,  & ~ \mbox{when} ~ t < 2s-3n.
 \end{cases}\]

\item[(e)] If $I=\left \langle  (x+1)^s+u (x+1)^t (2h_1(x)+(x+1)^lh_2(x)) \right \rangle$, where $n < s \leq 2n-1$, $0 \leq t < s-n$, $l > s-n-t$ and $h_1(x)$, $h_2(x)$ are units in $R_n^{\prime \prime}$, then
 \[A(I)= \left \langle  (x+1)^{s-t}+ u(2h_1(x)+(x+1)^{l}h_2(x)), ~u(x+1)^{2n-s} \right \rangle.\]
  \end{enumerate}
 \end{theorem}

\begin{corollary} Let $C$ be a non-trivial principal negacyclic code of length $n$ over $R$. Then the dual $C^{\perp}$ of $C$ is given  as follows:
\begin{enumerate}
\item[(a)] If $C$= $\left \langle  u(x+1)^m \right \rangle$, where $1 \leq m \leq 2n-1$, then
 $C^\perp=\left \langle  (x+1)^{2n-m},~u \right \rangle$.
 \item[(b)] If $C=\left \langle  (x+1)^s+u (x+1)^t h(x) \right \rangle$, where $1 \leq s \leq n-1$, $0 \leq t < s$ and $h(x)$ is either zero or a unit in $R_n^{\prime \prime}$, then
\[C^\perp=\begin{cases}  \left \langle  (x+1)^{2n-s}+ ux^{s-t}(x+1)^{2n-2s+t} h\left(\frac{1}{x}\right) \right \rangle, & ~\mbox{when}~ t < 2s-n,\\
 \left \langle  (x+1)^{2n-s}+2ux^{n+s-t-\frac{n}{2}}(x+1)^{n-2s+t} h\left(\frac{1}{x}\right) \right \rangle,  & ~ \mbox{when}~ t \geq 2s-n.
 \end{cases}\]
 \item[(c)] If $C=\left \langle  (x+1)^{s}+u(x+1)^th(x) \right \rangle$, where $n \leq s \leq 2n-1$, $0 \leq t \leq n-1$ and $h(x)$ is either zero  or a unit in $R_n^{\prime \prime}$, then
 \[C^\perp=\begin{cases}
 \left \langle (x+1)^{s-t}+ux^{s-t}h\left(\frac{1}{x}\right), ~u(x+1)^{2n-s}\right \rangle, & ~\mbox{when}~t < 2s-2n, \\
 \left \langle (x+1)^{2n-s}+ux^{s-t}(x+1)^{2n-2s+t}h\left(\frac{1}{x}\right) \right \rangle,  & ~ \mbox{when} ~ t \geq 2s-2n.
 \end{cases}\]
 \item[(d)] If $C=\left \langle  (x+1)^{s}+2u(x+1)^th(x) \right \rangle $, where $n < s \leq  2n-1$ and $0 \leq t < s-n$ and $h(x)$ is a unit in $R_n^{\prime \prime}$, then
 \[C^\perp=\begin{cases} \left \langle  (x+1)^{2n-s}+ux^{s-n-t}(x+1)^{3n-2s+t} \left(1+\left(\frac{1}{x}+1\right)^{\frac{n}{2}}\right)h\left(\frac{1}{x}\right) \right \rangle, &~\mbox{when}~ 2s-3n \leq t < s-n,\\
 \left \langle  (x+1)^{n-t}, ~u(x+1)^{2n-s}  \right \rangle,  & ~ \mbox{when} ~ t < 2s-3n.
 \end{cases}\]
\item[(e)] If $C=\left \langle  (x+1)^s+u (x+1)^t (2h_1(x)+(x+1)^lh_2(x)) \right \rangle$, where $n < s \leq 2n-1$, $0 \leq t < s-n$, $l >s-n-t$ and $h_1(x),~ h_2(x)$ are units in $R_n^{\prime \prime}$, then
 \[C^\perp= \left \langle  (x+1)^{s-t}+u\left(2x^{s-t}h_1\left(\frac{1}{x}\right)+ x^{s-t-l}(x+1)^{l}h_2\left(\frac{1}{x}\right)\right), ~u(x+1)^{2n-s} \right \rangle.\] \end{enumerate}
\end{corollary}

\begin{theorem}\label{duals2} Let $I$ be a non-principal ideal and $A(I)$ be its annihilator in $R_n$.
\begin{enumerate}
 \item[(a)] If $I=\left \langle  (x+1)^{s}+ u(x+1)^th(x), ~u(x+1)^m \right \rangle $, where $1 \leq s < n$, $0 \leq t < m < s$ and $ h(x)$ is either zero or a unit in $R_n^{\prime \prime}$, then
 \[A(I)=\left \langle  (x+1)^{2n-m}+u(x+1)^{2n-m-s+t}h(x),~u(x+1)^{2n-s} \right \rangle.\]
 \item[(b)] If $I=\left \langle  (x+1)^{s}+u(x+1)^th(x) ,~u(x+1)^m \right \rangle $, where $n \leq s < 2n$, $0 \leq t < n$, $ 1+t \leq m < T$ and $ h(x)$ is either zero or a unit in $R_n^{\prime \prime}$, then
 \[A(I)=\begin{cases}
 \left \langle  (x+1)^{2n-m}+u(x+1)^{2n-m-s+t}h(x),~u(x+1)^{2n-s} \right \rangle, &~\mbox{when}~ t \geq s+m-2n,\\
 \left \langle  (x+1)^{s-t}+uh(x),~u(x+1)^{2n-s}  \right \rangle,  & ~ \mbox{when} ~ t < s+m-2n.
 \end{cases}\]
 
\item[(c)] If $I=\left \langle  (x+1)^s+2u (x+1)^t h(x), ~u(x+1)^m \right \rangle$, where $n < s \leq 2n-1$, $0 \leq t < s-n$, $1+t \leq m < n$ and $h(x)$ is a unit in $R_n^{\prime \prime}$, then
 \[A(I)= \left \langle  (x+1)^{2n-m},~u(x+1)^{2n-s} \right \rangle.\]

 \item[(d)] If $I=\left \langle  (x+1)^{s}+2u(x+1)^th(x) ,~2u(x+1)^{m_1} \right \rangle $, where $n < s \leq 2n-1$, $0 \leq t < s-n$, $1+t \leq m_1 < T_1$ and $h(x)$ is a unit in $R_n^{\prime \prime}$, then
 \[A(I)=\begin{cases}
 \left \langle  (x+1)^{n-m_1}+u(x+1)^{2n-m-s+t}(1+(x+1)^{\frac{n}{2}})h(x),~u(x+1)^{2n-s} \right \rangle, &~\mbox{when}~ t \geq s+m_1-2n,\\
 \left \langle  (x+1)^{n-t},~u(x+1)^{2n-s}  \right \rangle,  & ~ \mbox{when} ~ t < s+m_1-2n.
 \end{cases}\]
\item[(e)] If $I=\left \langle  (x+1)^s+u (x+1)^t (2h_1(x)+(x+1)^lh_2(x)),~u(x+1)^m \right \rangle$, where $n < s \leq 2n-1$, $0 \leq t < s-n$, $t+l+1 \leq m  < n$ and $h_1(x)$, $h_2(x)$ are units in $R_n^{\prime \prime}$ then
 \[A(I)=\begin{cases}
 \left \langle  (x+1)^{2n-m}+u(x+1)^{2n-m-s+t+l}h_2(x),~u(x+1)^{2n-s} \right \rangle, &~\mbox{when}~ t+l \geq s+m-2n,\\
 \left \langle  (x+1)^{2n-l},~u(x+1)^{2n-s}  \right \rangle,  & ~ \mbox{when} ~ t+l < s+m-2n.
 \end{cases}\]
\item[(f)] If $I=\left \langle  (x+1)^s+u (x+1)^t (2h_1(x)+(x+1)^lh_2(x)),~2u(x+1)^{m_1} \right \rangle$, where $n < s \leq 2n-1$, $0 \leq t < s-n$, $t < m_1 < T_1$ and $h_1(x)$, $h_2(x)$ are units in $R_n^{\prime \prime}$, then
    \[A(I)=\begin{cases}
 \left \langle  (x+1)^{n-m_1}+u(x+1)^{n-m_1-s+t+l}h^\prime(x), ~u(x+1)^{2n-s} \right \rangle, &~\mbox{when} ~t+l \geq s+m_1-n,\\
 \left \langle  (x+1)^{2n-l}, ~u(x+1)^{2n-s} \right \rangle, &~\mbox{when} ~t+l < s+m_1-n, \end{cases}\] where $h^\prime(x)=(h_1(x)(x+1)^{n-l}+h_2(x)(1+(x+1)^{\frac{n}{2}}))$.
 \end{enumerate}
 \end{theorem}

\begin{corollary} Let $C$ be a non-principal negacyclic code of length $n$ over $R$. Then the dual $C^{\perp}$ of $C$ is given  as follows:
\begin{enumerate}
\item[(a)] If $C=\left \langle  (x+1)^{s}+u(x+1)^th(x) ,~u(x+1)^m \right \rangle $, where $1 \leq s < n$ and $0 \leq t < m < s$ and $h(x)$ is either zero or a unit in $R_n^{\prime \prime}$, then
     \[C^\perp=\left \langle  (x+1)^{2n-m}+ux^{s-t}(x+1)^{2n-m-s+t}h\left(\frac{1}{x}\right), ~u(x+1)^{2n-s} \right \rangle.\]
 \item[(b)] If $C=\left \langle  (x+1)^{s}+u(x+1)^th(x) ,~u(x+1)^m \right \rangle $, where $n \leq s < 2n$ and $0 \leq t < n$, $ 1+t \leq m < T$ and $h(x)$ is either zero or a unit in $R_n^{\prime \prime}$, then
 \[C^\perp=\begin{cases}
 \left \langle  (x+1)^{2n-m}+ ux^{s-t}(x+1)^{2n-m-s+t} h\left(\frac{1}{x}\right), ~u(x+1)^{2n-s} \right \rangle, &~\mbox{when}~ t \geq s+m-2n,\\
 \left \langle  (x+1)^{s-t}+ux^{s-t}h\left(\frac{1}{x}\right),~u(x+1)^{2n-s}  \right \rangle,  & ~ \mbox{when} ~ t < s+m-2n.
 \end{cases}\]
\item[(c)] If $C=\left \langle  (x+1)^s+2u (x+1)^t h(x), ~u(x+1)^m \right \rangle$, where $n < s \leq 2n-1$, $0 \leq t < s-n$, $1+t \leq m < n$ and $h(x)$ is a unit in $R_n^{\prime \prime}$, then
 \[C^\perp= \left \langle  (x+1)^{2n-m},~u(x+1)^{2n-s} \right \rangle.\]
 \item[(d)] If $C=\left \langle  (x+1)^{s}+2u(x+1)^th(x) ,~2u(x+1)^{m_1} \right \rangle $, where $n < s \leq 2n-1$, $0 \leq t < s-n$, $1+t \leq m_1 < T_1$ and $h(x)$ is a unit in $R_n^{\prime \prime}$, then
 \[C^\perp=\begin{cases}
 \left \langle  (x+1)^{n-m_1}+ux^{s-n-t}(x+1)^{2n-m_1-s+t}h^\prime(x), ~u(x+1)^{2n-s} \right \rangle, &~\mbox{when}~ t \geq s+m_1-2n,\\
 \left \langle  (x+1)^{n-t},~u(x+1)^{2n-s}  \right \rangle,  & ~ \mbox{when} ~ t < s+m_1-2n,
 \end{cases}\] where $h^\prime(x)=(1+\left(\frac{1}{x}+1\right)^{\frac{n}{2}}) h\left(\frac{1}{x}\right)$.
\item[(e)] If $C=\left \langle  (x+1)^s+u (x+1)^t (2h_1(x)+(x+1)^lh_2(x)),~u(x+1)^m \right \rangle$, where $n < s \leq 2n-1$, $0 \leq t < s-n$, $l > s-n-t$, $t+l+1 \leq m  < n$ and $h_1(x),~ h_2(x)$ are units in $R_n^{\prime \prime}$, then
    \[C^\perp=\begin{cases}
 \left \langle (x+1)^{2n-m}+ux^{s-t-l}(x+1)^{2n-m-s+t+l}h_2\left(\frac{1}{x}\right), ~u(x+1)^{2n-s} \right \rangle, &~\mbox{when}~ t+l \geq s+m-2n,\\
 \left \langle  (x+1)^{2n-l},~u(x+1)^{2n-s}  \right \rangle,  & ~ \mbox{when} ~ t+l < s+m-2n.
 \end{cases}\]
\item[(f)] If $C=\left \langle  (x+1)^s+u (x+1)^t (2h_1(x)+(x+1)^lh_2(x)),~2u(x+1)^{m_1} \right \rangle$, where $n < s \leq 2n-1$, $0 \leq t < s-n$, $l > s-n-t$, $t < m_1 < T_1$ and $h_1(x),~ h_2(x)$ are units in $R_n^{\prime \prime}$, then
    \[C^\perp=\begin{cases}
 \left \langle  (x+1)^{n-m_1}+ux^{s-t-l}(x+1)^{n-m_1-s+t+l} h^{\prime \prime}(x), ~u(x+1)^{2n-s} \right \rangle, &~\mbox{when} ~t+l \geq s+m_1-n,\\
 \left \langle  (x+1)^{2n-l}, ~u(x+1)^{2n-s} \right \rangle, &~\mbox{when} ~t+l < s+m_1-n, \end{cases}\] where $h^{\prime \prime}(x)=h_1\left(\frac{1}{x}\right) (\frac{1}{x}+1)^{n-l}+h_2\left(\frac{1}{x}\right) \left(1+\left(\frac{1}{x}+1\right)^{\frac{n}{2}}\right)$.
 \end{enumerate}
\end{corollary}

 Now we consider the cardinality of the negacyclic code $C$ over $R$. If $|$Tor $(C)|$ and $|$Res $(C)|$ are known, then $|C|$ can be computed by $|C|=|Tor (C)||Res (C)|$. The following theorem gives the Tor$(C)$ and Res$(C)$ of $C$ in each case.

\begin{lemma}\label{z4} \cite[Theorem 6.10]{dinh1} The negacyclic codes of length $n$ over $\mathbb{Z}_4$ are precisely the ideals $\left \langle (x+1)^i \right \rangle$, $0 \leq i \leq 2n$, of $\frac{\mathbb{Z}_4[x]}{\left \langle x^n+1 \right \rangle}$.
\end{lemma}
\begin{lemma}\label{z4 size} \cite[Theorem 6.14]{dinh1} Let $\mathcal{C}= \left \langle (x+1)^i \right \rangle$, $0 \leq i \leq 2n$ be a negacyclic code of length $n$ over $\mathbb{Z}_4$. Then $|\mathcal{C}|=2^{2n-i}$.
\end{lemma}

\begin{theorem} \label{tor}
Let $C$ be a non-trivial negacyclic code of length $n$ over $R$.
\begin{itemize}
\item[(a)] If  $C=\left \langle  u(x+1)^{m}\right \rangle$, where $0 \leq m \leq 2n$, then $\mbox{Res}(C)=\langle 0  \rangle$ and $\mbox{Tor}(C)=\langle (x+1)^m \rangle$.
\item[(b)] If $C=\left \langle  (x+1)^s+u (x+1)^t h(x) \right \rangle$, where $0 \leq s \leq n-1$, $0 \leq t < s$ and $h(x)$ either zero or a unit in $R_n^{\prime \prime}$, then $\mbox{Res}(C)=\mbox{Tor}(C)=\left \langle (x+1)^s \right \rangle$.
\item[(c)] If  $C=\left \langle  (x+1)^s+u (x+1)^t h(x) \right \rangle$, where $n \leq s \leq 2n-1$, $0 \leq t < n$ and $h(x)$ either zero or a unit in $R_n^{\prime \prime}$, then $\mbox{Res}(C)=\left \langle (x+1)^s  \right \rangle$ and $\mbox{Tor}(C)=\left \langle (x+1)^T  \right \rangle$.
    \item[(d)] If  $C=\left \langle  (x+1)^s+2u (x+1)^t h(x) \right \rangle$, where $n < s \leq 2n-1$, $0 \leq t < s-n$, $h(x)$ is a unit in $R_n^{\prime \prime}$, then $\mbox{Res}(C)=\left \langle (x+1)^s  \right \rangle$ and $\mbox{Tor}(C)=\left \langle (x+1)^T  \right \rangle$.
\item[(e)] If $C=\left \langle  (x+1)^s+u (x+1)^t (2h_1(x)+(x+1)^lh_2(x)) \right \rangle$, where $n < s \leq 2n-1$, $0 \leq t < s-n$, $l > s-n-t$, $h_1(x)$ and $h_2(x)$ are units in $R_n^{\prime \prime}$, then $\mbox{Res}(C)=\left \langle (x+1)^s  \right \rangle$ and $\mbox{Tor}(C)=\left \langle (x+1)^T  \right \rangle$, where $T=\mbox{min}\{s,2n-s+t+l\}$.
\item[(f)]  If $C=\left \langle  (x+1)^s+u (x+1)^t h(x), ~u(x+1)^m \right \rangle$, where $1 \leq s \leq n-1$, $0 \leq t < m < s$, $h(x)$ is either zero or a unit in $R_n^{\prime \prime}$, then $\mbox{Res}(C)=\left \langle (x+1)^s  \right \rangle$ and $\mbox{Tor}(C)=\left \langle (x+1)^m  \right \rangle$.
\item[(g)]  If $C=\left \langle  (x+1)^s+u (x+1)^t h(x), ~u(x+1)^m \right \rangle$, where $n \leq s \leq 2n-1$, $0 \leq t < n$, $1+t \leq m < T$, $h(x)$ is either zero or a unit in $R_n^{\prime \prime}$, then $\mbox{Res}(C)=\left \langle (x+1)^s  \right \rangle$ and $\mbox{Tor}(C)=\left \langle (x+1)^m  \right \rangle$.
\item[(h)] If $C=\left \langle  (x+1)^s+2u (x+1)^t h(x), ~u(x+1)^m \right \rangle$, where $n < s \leq 2n-1$, $0 \leq t < s-n$, $1+t \leq m < n$, $h(x)$ is a unit in $R_n^{\prime \prime}$, then $\mbox{Res}(C)=\left \langle (x+1)^s  \right \rangle$  and $\mbox{Tor}(C)=\left \langle (x+1)^m  \right \rangle$.
\item[(i)]  If $C=\left \langle  (x+1)^s+2u (x+1)^t h(x), ~2u(x+1)^{m_1} \right \rangle$, where $n < s \leq 2n-1$, $0 \leq t < s-n$, $1+t \leq m_1 < T_1$, $h(x)$ is a unit in $R_n^{\prime \prime}$, then $\mbox{Res}(C)=\left \langle (x+1)^s  \right \rangle$ and $\mbox{Tor}(C)=\left \langle (x+1)^{n+m_1}  \right \rangle$.
\item[(j)] If $C=\left \langle  (x+1)^s+u (x+1)^t (2h_1(x)+(x+1)^lh_2(x)), ~u(x+1)^m \right \rangle$,  where $n < s \leq 2n-1$, $0 \leq t < s-n$, $l > s-n-t$, $l+t < m  < n$ and $h_1(x), h_2(x)$ are units in $R_n^{\prime \prime}$, then $\mbox{Res}(C)=\left \langle (x+1)^s  \right \rangle$ and $\mbox{Tor}(C)=\left \langle (x+1)^m  \right \rangle$.
\item[(k)] If $C=\left \langle  (x+1)^s+u (x+1)^t (2h_1(x)+(x+1)^lh_2(x)), ~2u(x+1)^{m_1} \right \rangle$,  where $n < s \leq 2n-1$, $0 \leq t < s-n$, $l > s-n-t$  $t < m_1  < T_1$ and $h_1(x)$, $h_2(x)$ are units in $R_n^{\prime \prime}$, then $\mbox{Res}(C)=\left \langle (x+1)^s  \right \rangle$  and $\mbox{Tor}(C)=\left \langle (x+1)^{n+m_1}  \right \rangle$.
 \end{itemize}
\end{theorem}

\begin{theorem}\label{size of C}
Let $C$ be a negacyclic code of length $n$ over $R$. Then
\begin{enumerate}
\item[(a)] If $C=\left \langle  u(x+1)^{m}\right \rangle$, where $0 \leq m \leq 2n$, then $|C|=2^{2n-m}$.
\item[(b)] If $C=\left \langle  (x+1)^s+u (x+1)^th(x) \right \rangle$, where $0 \leq s \leq n-1$, $0 \leq t < s$, then $|C|=4^{2n-s}.$
\item[(c)] If $C=\left \langle  (x+1)^s+u (x+1)^th(x) \right \rangle$, where $n \leq s \leq 2n-1$, $0 \leq t < n$,  then \[|C|=\begin{cases}
    2^{2n-t} & \mbox{if}~~ 0 \leq t < 2s-2n,\\
    4^{2n-s} & \mbox{if}~~t \geq 2s-2n.
    \end{cases}\]
\item[(d)]  If  $C=\left \langle  (x+1)^s+2u (x+1)^t h(x) \right \rangle$, where $n < s \leq 2n-1$, $0 \leq t < s-n$, then \[|C|=\begin{cases}
    2^{n-t} & \mbox{if}~~ 0 \leq t < 2s-3n,\\
    4^{2n-s} & \mbox{if}~~2s - 3n \leq t < s-n.
    \end{cases}\]
\item[(e)] If $C=\left \langle  (x+1)^s+u (x+1)^t (2h_1(x)+(x+1)^lh_2(x)) \right \rangle$, where $n < s \leq 2n-1$, $0 \leq t < s-n$ and $l > s-n-t$, then\[|C|=\begin{cases}
    2^{2n-l-t} & \mbox{if}~~ s-n < l+t < 2s-2n,\\
    4^{2n-s} & \mbox{if}~~l+t \geq 2s-2n.
    \end{cases}\]
\item[(f)]  If $C=\left \langle  (x+1)^s+u (x+1)^t h(x), ~u(x+1)^m \right \rangle$, where $1 \leq s \leq n-1$, $0 \leq t < m < s$, $h(x)$ is either zero or a unit in $R_n^{\prime \prime}$, then  $|C|=2^{4n-(m+s)}$.
\item[(g)] If $C=\left \langle  (x+1)^s+u (x+1)^t h(x), ~u(x+1)^m \right \rangle$, where $n \leq s \leq 2n-1$, $0 \leq t < n$, $1+t \leq m < T$, $h(x)$ is either zero or a unit in $R_n^{\prime \prime}$, then $|C|=2^{4n-(m+s)}$.
\item[(h)] If $C=\left \langle  (x+1)^s+2u (x+1)^t h(x), ~u(x+1)^m \right \rangle$, where $n < s \leq 2n-1$, $0 \leq t < s-n$, $1+t \leq m < n$, $h(x)$ is a unit in $R_n^{\prime \prime}$, then $|C|=2^{4n-(m+s)}$.
\item[(i)]  If $C=\left \langle  (x+1)^s+2u (x+1)^t h(x), ~2u(x+1)^{m_1} \right \rangle$, where $n < s \leq 2n-1$, $0 \leq t < s-n$, $1+t \leq m_1 < T_1$, $h(x)$ is a unit in $R_n^{\prime \prime}$, then $|C|=2^{3n-(m_1+s)}$.
\item[(j)] If $C=\left \langle  (x+1)^s+u (x+1)^t (2h_1(x)+(x+1)^lh_2(x)), ~u(x+1)^m \right \rangle$,  where $n < s \leq 2n-1$, $0 \leq t < s-n$, $l >s-n-t$, $l+t < m  < n$ and $h_1(x), h_2(x)$ are units in $R_n^{\prime \prime}$, then $|C|=2^{4n-(m+s)}$.
\item[(k)] If $C=\left \langle  (x+1)^s+u (x+1)^t (2h_1(x)+(x+1)^lh_2(x)), ~2u(x+1)^{m_1} \right \rangle$,  where $n < s \leq 2n-1$, $0 \leq t < s-n$, $l > s-n-t$, $t < m_1  < T_1$ and $h_1(x)$, $h_2(x)$ are units in $R_n^{\prime \prime}$, then $|C|=2^{3n-(m_1+s)}$.
\end{enumerate}
\end{theorem}

The following Theorem gives the self-orthogonal negacyclic codes over $R$.

\begin{theorem}\label{selfortho} The self-orthogonal negcyclic codes of length $n$ over $R$ are:
\begin{enumerate}
\item[(a)] $C=\left \langle  u(x+1)^{m}\right \rangle$, where $0 \leq m \leq 2n$.

 \item[(b)] $C=\left \langle  (x+1)^s+u (x+1)^t h(x) \right \rangle$, where  $n \leq s \leq 2n-1$, $0 \leq t < n$ and $h(x)$ is either zero or a unit in $R_n^{\prime \prime}$.

\item[(c)]    $C=\left \langle  (x+1)^s+2u (x+1)^t h(x) \right \rangle$, where $n < s \leq 2n-1$, $0 \leq t < s-n$ and $h(x)$ is a unit in $R_n^{\prime \prime}$.

\item[(d)]  $C=\left \langle  (x+1)^s+u (x+1)^t (2h_1(x)+(x+1)^lh_2(x)) \right \rangle$,  where $n < s \leq 2n-1$, $0 \leq t < s-n$, $l > s-n-t$ and $h_1(x)$, $h_2(x)$ are units in $R_n^{\prime \prime}$,

\item[(e)]  $C=\left \langle  (x+1)^s+u (x+1)^t h(x), ~u(x+1)^m \right \rangle$, where $n \leq s \leq 2n-1$, $0 \leq t < m < T$, $h(x)$ is either zero or a unit in $R_n^{\prime \prime}$, and $s+m \geq 2n$.

\item[(f)]  $C=\left \langle  (x+1)^s+2u (x+1)^t h(x), ~u(x+1)^m \right \rangle$, where  $n < s \leq 2n-1$, $0 \leq t < s-n$, $1+t \leq m < n$, $h(x)$ is a unit in $R_n^{\prime \prime}$, and $s+m \geq 2n$.

\item[(g)]  $C=\left \langle  (x+1)^s+2u (x+1)^t h(x), ~2u(x+1)^{m_1} \right \rangle$, where $n < s \leq 2n-1$, $0 \leq t < s-n$, $1+t \leq m_1 < T_1$, $h(x)$ is a unit in $R_n^{\prime \prime}$, and $s+m_1 \geq n$.

\item[(h)] $C=\left \langle  (x+1)^s+u (x+1)^t (2h_1(x)+(x+1)^lh_2(x)), ~u(x+1)^m \right \rangle$, where  $n < s \leq 2n-1$, $0 \leq t < s-n$, $l >s-n-t$, $t+l < m  < n$, $h_1(x)$, $h_2(x)$ are  units in $R_n^{\prime \prime}$, and $s+m \geq 2n$.

\item[(i)] $C=\left \langle  (x+1)^s+u (x+1)^t (2h_1(x)+(x+1)^lh_2(x)), ~2u(x+1)^{m_1} \right \rangle$, where $n < s \leq 2n-1$, $0 \leq t < s-n$, $l >s-n-t$, $t < m_1  < T_1$, $h_1(x)$, $h_2(x)$ are units in $R_n^{\prime \prime}$, and $s+m_1 \geq n$.
\end{enumerate}
\end{theorem}

    In the following theorem we list all self-dual negacylic codes of length $n$ over $R$.

\begin{theorem}\label{selfdual} The only self-dual  negacyclic codes of length $n$ over $R$ are:
\begin{enumerate}
\item[(a)]  $C=\left \langle  (x+1)^{n}+u (x+1)^th(x) \right \rangle$, where $ t \geq 0$ and $h(x)$ is either zero or a unit in $R_n^{\prime \prime}$.
\item[(b)] $C=\left \langle  (x+1)^{s}+u h(x) \right \rangle$, where  $n < s \leq 2n-1$, $h(x)$ is either zero or a unit in $R_n^{\prime \prime}$.
\item[(c)]  $C=\left \langle  (x+1)^s+u (x+1)^t h(x), ~u(x+1)^m \right \rangle$, where $n \leq s \leq 2n-1$, $0 \leq t < m < T$, $h(x)$ is either zero or a unit in $R_n^{\prime \prime}$ and $s+m = 2n$.
\item[(d)]  $C=\left \langle  (x+1)^s+2u (x+1)^t h(x), ~u(x+1)^m \right \rangle$, where  $n < s \leq 2n-1$, $0 \leq t < s-n$, $1+t \leq m < n$, $h(x)$ is  a unit in $R_n^{\prime \prime}$ and $s+m = 2n$.
\item[(e)]  $C=\left \langle  (x+1)^s+2u (x+1)^t h(x), ~2u(x+1)^{m_1} \right \rangle$, where $n < s \leq 2n-1$, $0 \leq t < s-n$, $1+t \leq m_1 < T_1$, $h(x)$ is a unit in $R_n^{\prime \prime}$ and $s+m_1 = n$.
\item[(f)] $C=\left \langle  (x+1)^s+u (x+1)^t (2h_1(x)+(x+1)^lh_2(x)), ~u(x+1)^m \right \rangle$, where   $n < s \leq 2n-1$, $0 \leq t < s-n$, $l > s-n-t$, $t+l < m  < n$, $h_1(x)$, $h_2(x)$ are units in $R_n^{\prime \prime}$, and $s+m = 2n$.
\item[(g)] $C=\left \langle  (x+1)^s+u (x+1)^t (2h_1(x)+(x+1)^lh_2(x)), ~2u(x+1)^{m_1} \right \rangle$, where $n < s \leq 2n-1$, $0 \leq t < s-n$, $l > s-n-t$, $t < m_1  < T_1$, $h_1(x)$, $h_2(x)$ are units in $R_n^{\prime \prime}$, and $s+m_1 = n$.
    \end{enumerate}
\end{theorem}
\begin{example}
 For $n=2$, $R_2=\frac{R[x]}{\langle x^2+1 \rangle}$ has $24$ ideals (negacyclic codes of length $2$ over $R$), out of which $7$ are self-dual ($C^\textbf{*}$) and $8$ are self-orthogonal ($C^\textbf{\dag}$). They are listed below:
\end{example}
\begin{tabular}{|c| c|c|}
\hline
Negacyclic code $C$  & Annihilator $A(C)$ & Size of $C$  \\[1mm]
\hline
$C_1=\langle 0 \rangle$ & $C_2$ & $1$ \\[1mm]
\hline
$C_2=\langle 1 \rangle$ & $C_1$ & $256$\\[1mm]
\hline
$C_3=\langle u \rangle$  & $C_3^\textbf{*}$ & $16$\\[1mm]
\hline
$C_4=\langle u(x+1) \rangle$ & $C_{22}^\textbf{\dag}$ & $8$\\[1mm]
\hline
$C_5=\langle u(x+1)^2 \rangle$ & $C_{19}^\textbf{\dag}$ & $4$\\[1mm]
\hline
$C_6=\langle u(x+1)^3 \rangle$ & $C_{18}^\textbf{\dag}$ & $2$\\[1mm]
\hline
$C_7=\langle (x+1) \rangle$ & $C_{9}$ & $64$\\[1mm]
\hline
$C_8=\langle (x+1)^2 \rangle$  & $C_{8}^\textbf{*}$ & $16$\\[1mm]
\hline
$C_9=\langle (x+1)^3 \rangle$ & $C_{7}^\textbf{\dag}$ & $4$\\[1mm]
\hline
$C_{10}=\langle (x+1)+u \rangle$ & $C_{16}$ & $64$\\[1mm]
\hline
$C_{11}=\langle (x+1)^2+u \rangle$ & $C_{11}^\textbf{*}$ & $16$\\[1mm]
\hline
$C_{12}=\langle (x+1)^2+u(x+1) \rangle$ & $C_{12}^\textbf{*}$ & $16$\\[1mm]
\hline
$C_{13}=\langle (x+1)^2+u(1+(x+1)) \rangle$ & $C_{13}^\textbf{*}$ & $16$\\[1mm]
\hline
$C_{14}=\langle (x+1)^3+u \rangle$ & $C_{14}^\textbf{*}$ & $16$\\[1mm]
\hline
$C_{15}=\langle (x+1)^3+u(x+1) \rangle$ & $C_{21}^\textbf{\dag}$ & $8$ \\[1mm]
\hline
$C_{16}=\langle (x+1)^3+2u \rangle$ & $C_{10}^\textbf{\dag}$ & $4$\\[1mm]
\hline
$C_{17}=\langle (x+1), ~u \rangle$ & $C_{6}$ & $128$\\[1mm]
\hline
$C_{18}=\langle (x+1)^2, ~u   \rangle$ & $C_{5}$ & $64$\\[1mm]
\hline
$C_{19}=\langle  (x+1)^2,~u(x+1) \rangle$ & $C_{24}$ & $32$\\[1mm]
\hline
$C_{20}=\langle  (x+1)^2+u, ~u(x+1)  \rangle$ & $C_{15}$ & $32$\\[1mm]
\hline
$C_{21}=\langle  (x+1)^3, ~u  \rangle$ & $C_{4}$ & $32$\\[1mm]
\hline
$C_{22}=\langle  (x+1)^3,~u(x+1) \rangle$ & $C_{23}^\textbf{*}$ & $16$\\[1mm]
\hline
$C_{23}=\langle  (x+1)^3,~2u \rangle$ & $C_{20}^\textbf{\dag}$ & $8$\\[1mm]
\hline
$C_{24}=\langle  (x+1)^3+2u,~u(x+1) \rangle$ & $C_{23}^\textbf{\dag}$ & $16$\\[1mm]
\hline
\end{tabular}

\section{Conclusion}
In this paper we have studied some structural properties of negacyclic codes over the ring $R= \mathbb{Z}_4+u\mathbb{Z}_4$ for both odd length and length $2^k$. Since $R$ is not a principal ideal ring, the ideal structure of $\frac{R[x]}{\left \langle x^n+1\right \rangle}$ is not as well understood.  We have studied some properties of principally generated negacyclic codes of odd length over $R$. We have obtained the structure of ideals of $\frac{R[x]}{\left \langle x^{2^k}+1\right \rangle}$. We have also studied self-orthogonal and self-dual negacyclic codes length $2^k$ over $R$.

\end{document}